\documentclass[letterpaper, 10 pt, conference]{ieeeconf}  % Comment this line out
\IEEEoverridecommandlockouts                              % This command is only
\usepackage{amsmath} % assumes amsmath package installed
\usepackage{amssymb}  % assumes amsmath package installed
\usepackage{color}
\usepackage{cite}
\usepackage{epsfig} % for postscript graphics files
\usepackage{graphics} % for pdf, bitmapped graphics files
\usepackage{graphicx}
\usepackage{ifpdf}
\usepackage{mathptmx} % assumes new font selection scheme installed
\usepackage{subfigure}
\usepackage{times} % assumes new font selection scheme installed
\usepackage{times} % assumes new font selection scheme installed

\newtheorem{theorem}{Theorem}

\title{\LARGE \bf
Spectral Analysis of Virus Spreading in Random Geometric Networks
}

\author{Victor M. Preciado and Ali Jadbabaie% <-this % stops a space
\thanks{This work was supported by ONR MURI N000140810747, and AFOR's complex networks program.}% <-this % stops a space
\thanks{The authors are with the Department of Electrical and Systems Engineering, University of Pennsylvania, 3451 Walnut Street, %Philadelphia, PA 19104, USA
        {\tt\small \{preciado,jadbabai\}@seas.upenn.edu}}%
}

\begin{document}

\maketitle
\thispagestyle{empty}
\pagestyle{empty}

%%%%%%%%%%%%%%%%%%%%%%%%%%%%%%%%%%%%%%%%%%%%%%%%%%%%%%%%%%%%%%%%%%%%%%%%%%%%%%%%
\begin{abstract}

In this paper, we study the dynamics of a viral spreading process in random geometric graphs (RGG). The spreading of the viral process we consider in this paper is closely related with the eigenvalues of the adjacency matrix of the graph. We deduce new explicit expressions for all the moments of the eigenvalue distribution of the adjacency matrix as a function of the spatial density of nodes and the radius of connection. We apply these expressions to study the behavior of the viral infection in an RGG. Based on our results, we deduce an analytical condition that can be used to design RGG's in order to tame an initial viral infection. Numerical simulations are in accordance with our analytical predictions.

\end{abstract}

%%%%%%%%%%%%%%%%%%%%%%%%%%%%%%%%%%%%%%%%%%%%%%%%%%%%%%%%%%%%%%%%%%%%%%%%%%%%%%%%
\section{Introduction}

The analysis of spreading processes in large-scale complex networks is a
fundamental dynamical problem in network science. The relationship between
the dynamics of epidemic/information spreading and the structure of the
underlying network is crucial in many practical cases, such as the spreading
of worms in a computer network, viruses in a human population, or rumors in
a social network. Several papers approached different facets of the virus
spreading problem. A rigorous analysis of epidemic spreading in a
finite one-dimensional linear network was developed by Durrett and Liu in 
\cite{DL88}. In \cite{WCWF03}, Wang et al. derived a sufficient condition to
tame an epidemic outbreak in terms of the spectral radius of the adjacency
matrix of the underlying graph. Similar results were derived by Ganesh
et al. in \cite{GMT05}, establishing a connection between the behavior of a
viral infection and the eigenvalues of the adjacency matrix of the network.

In this paper, we study the dynamics of a viral spreading in an important type of proximity
networks called Random Geometric Graphs (RGG). RGG's consist of a set of
vertices randomly distributed in a given spatial region with edges connecting pairs of nodes that
are within a given distance $r$ from each other (also called \emph{connectivity radius}). In this paper, we
derive new explicit expressions for the expected spectral
moments of the random adjacency matrix associated to an RGG. Our results allow us to derive analytical conditions under which an RGG is well-suited to tame an infection in the network.

The paper is structured as follows. In Section II, we describe random
geometric graphs and introduce several useful results concerning their
structural properties. We also present the spreading model in \cite{WCWF03}
and review an important result that relates the behavior of an initial infection
with the spectral radius of the adjacency matrix. In Section III, we study
the eigenvalue spectrum of random geometric graphs. We derive explicit
expressions for the expected spectral moments in the case of one- and
two-dimensional RGG's. In Section IV, we use these expressions to study the spectral radius of RGG's. Our results allow us to
design RGG's with the objective of taming epidemic outbreaks. Numerical
simulations in Section IV validate our results.

\section{Virus Spreading in Random Geometric Graphs}

In this section, we briefly describe random geometric graphs and introduce
several useful results concerning their structural properties (see \cite%
{Pen03} for a thorough treatment). We then describe the spreading model
introduced in \cite{WCWF03} and show how to study the behavior of an
infection in the network from the point of view of the adjacency eigenvalues.

\subsection{\label{RGG}Random Geometric Graphs}

Consider a set of $n$ nodes, $V_{n}=\{v_{1},...,v_{n}\}$, respectively
located at random positions, $\chi _{n}=\{\mathbf{x}_{1},...,\mathbf{x}_{n}\}
$, where $\mathbf{x}_{i}$ are i.i.d. random vectors uniformly distributed on
the $d$-dimensional unit torus, $\mathbb{T}^{d}$. We use the torus for
convenience, to avoid boundary effects. We then connect two nodes $%
v_{i},v_{j}\in V_{n}$ if and only if $\left\Vert \mathbf{x}_{i}-\mathbf{x}%
_{j}\right\Vert \leq r $, where $r$ is the
so-called connectivity radius. In other words, a link exists between $v_{i}$
and $v_{j}$ if and only if $v_{j}$ lies inside the sphere of radius $r\left(
n\right) $ centered at $v_{i}$. We denote this spherical region by $%
S_{i}\left( r \right) $, and the resulting random geometric
graph by $G\left( \chi _{n};r \right) $. We define a \emph{walk }of length $%
k$ from $v_{0}$ to $v_{k}$ as an ordered set of (possibly repeated) vertices 
$\left( v_{0},v_{1},...,v_{k}\right) $ such that $v_{i}\sim v_{i+1},$ for $%
i=0,1,...,k-1$; if $\nu _{k}=\nu _{0}$ the walk is said to be \emph{closed}.

The \emph{degree} $d_{i}$ of a node $v_{i}$ is the number of edges
connected to it. In our case, the degrees are identical
random variables with expectation \cite{Pen03}:%
\begin{equation}
\mathbb{E}\left[ d_{i}\right] =nV^{\left( d\right) }r^{d},
\label{Expected Degree}
\end{equation}%
where $V^{\left( d\right) }$ is the volume of a $d$-dimensional unit sphere, 
$V^{\left( d\right) }=\pi ^{d/2}\left/ \Gamma \left( d/2+1\right) \right. $,
and $\Gamma \left( \cdot \right) $ is the Gamma function. The \emph{%
clustering coefficient} is a measure of the
number of triangles in a given graph, where a triangle is defined by the set
of edges $\left\{ \left( i,j\right) ,\left( j,k\right) ,\left( k,i\right)
\right\} $ such that $i\sim j\sim k\sim i$. For one- and two-dimensional
RGG's we can derive an explicit expression for the expected number of triangles, 
$\mathbb{E}\left[ t_{i}\right] $, touching a particular node $v_{i}$
(details are provided in Section III).

The \emph{adjacency matrix }of
an undirected graph $G,$ denoted by $A(G)=[a_{ij}]$, is defined entry-wise
by $a_{ij}=1$ if nodes $i$ and $j$ are connected, and $a_{ij}=0$ otherwise.
(Note that $a_{ii}=0$ for simple graphs.) Denote the eigenvalues of a $%
n\times n$ symmetric adjacency matrix $A(G)$ by $\lambda _{1}\leq ...\leq
\lambda _{n}$. The $k$\emph{-}th order
moment of the eigenvalue spectrum of $A(G)$ is defined as: 
\begin{equation*}
m_{k}(G)=\frac{1}{N}\sum_{i=1}^{n}\lambda _{i}^{k}
\end{equation*}%
(which is also called the $k$-th \emph{order spectral moment}).

We are interested in studying asymptotic properties of the sequence $G(\chi
_{n};r\left( n\right) )$ for some sequence $\left\{ r\left( n\right) :n\in 
\mathbb{N}\right\} $. In \cite{Pen03}, two particularly interesting regimes
are introduced: the \emph{thermodynamic limit} with $nr\left( n\right)
^{d}\rightarrow \alpha \in \left( 0,\infty \right) $, so that the expected
degree of a vertex tends to a constant, and the \emph{connectivity regime}
with $r\left( n\right) \rightarrow \gamma \left( \frac{\log n}{n}\right)
^{1/d}$ with a constant $\gamma $, so that the expected degree of the nodes
grows as $c\log n$. In this paper, we focus on studying the
spectral moments in the connectivity regime. In Section III, we
derive explicit expressions for the expected spectral moments of $G\left(
\chi _{n};r_{n}\right) $ for any network size $n$. We then use this
information to bound the spectral radius of the adjacency matrix of $G(\chi
_{n};r\left( n\right) )$.

\subsection{Spectral Analysis of Virus Spreading}

In this section, we briefly review an automaton model that describes the
dynamics of a viral infection in a specific network of interactions.
This model was proposed and analyzed in \cite{WCWF03}, where a connection between the growth of an initial infection in the
network and the spectral radius of the adjacency matrix was established. This model involves
several parameters. First, the infection rate $\beta $ represents the
probability of a virus at an infected node $i$ spreading to another
neighboring node $j$ during a time step. Also, we denote by $\delta $ the
probability of recovery of any infected node at each time step. For
simplicity, we consider $\beta $ and $\delta $ to be constants for all
the nodes in $G$. We also denote by $p_{i}\left[ k\right] $ the
probability that node $i$ is infected at time $k$. The evolution of the
probability of infection is modeled by means of the following system of
non-linear difference equation:%
\begin{equation}
p_{i}\left[ k+1\right] =[1-\prod_{j\in \mathcal{N}_{i}}\left( 1-\beta \,p_{j}%
\left[ k\right] \right) ]+\left( 1-\delta \right) p_{i}\left[ k\right] ,
\label{Epidemic Model}
\end{equation}%
for $i=1,...,n$, where $\mathcal{N}_{i}$ denotes the set of nodes connected
to node $i$. We are interested in studying the dynamics of the system for a low-density level of infection, i.e., $\beta \,p_{j}\left[ k\right] \ll 1$. In this regime, a sufficient condition for a small initial infection to die out
is  \cite{WCWF03}: 
\begin{equation}
\lambda _{\max }\left( A(G)\right) <\frac{\delta }{\beta }.
\label{Epidemic Conditions}
\end{equation}%
One can prove that (\ref{Epidemic Conditions}) is a sufficient condition for
local stability around the disease-free state. Thus, we can use condition (%
\ref{Epidemic Conditions}) to design networks with the objective of taming
initial low-density infections.

\section{Spectral Analysis of Random Geometric Graphs}

In this paper, we study the eigenvalue distribution of the random adjacency
matrix associated to $G(\chi _{n};r\left( n\right) )$\ for $n\rightarrow
\infty $. In this section, we characterize eigenvalue distribution using its sequence of spectral moments. In our derivations, we use an interesting graph-theoretical interpretation
of the spectral moments \cite{Big93}: \emph{the $k$-th spectral moment
of $G$ is proportional to the number of closed walks of length $k$ in $G$.} This result allows us to transform the algebraic problem of computing
spectral moments of the adjacency matrix into the combinatorial problem of
counting closed walks in the graph. In the following subsection, we compute the expected value of the number of
closed walks of length $k$ in $G(\chi _{n};r\left( n\right) )$.

\subsection{Spectral Moments of One-Dimensional RGG's}

As we mentioned above, we can compute the $k$-th spectral moment of a graph by counting the number of closed walks of length $k$. In the case of an RGG $G(\chi_{n};r\left( n\right) )$, this number is a random variable. In this subsection, we introduce a novel technique to compute the expected number of closed walks of length $k$. For clarity, we introduce our technique
for the first three expected spectral moments $k=1,2,3$. We then use these results to induce a general expression for higher-order moments in one-dimensional RGG's.

The first-order spectral moment is equal to the number of closed walks of
length $k=1$. Since $G(\chi _{n};r)$ is a simple graphs with no self-loops,
we have that $m_{1}\left( G(\chi _{n};r)\right) $ is a deterministic
quantity equal to $0$.

We now study the expected second moment, $\mathbb{E}\left[ m_{2}\left(
G(\chi _{n};r)\right) \right] $, by counting the number of closed walks of length two.
In simple graphs, the only possible closed walks of length two are those
that start at a given node $v_{i}$, visit a neighboring node $v_{j}\in 
\mathcal{N}_{i}$, and return back to $v_{i}$. Hence, the number of closed
walks of length two starting at $v_{i}$ is equal to $d_{i}$. Thus, from (\ref%
{Expected Degree}), we have%
\begin{equation*}
\mathbb{E}\left[ m_{2}\right] =\frac{1}{n}\sum_{i=1}^{n}\mathbb{E}\left[
d_{i}\right] =nV^{\left( d\right) }r^{d}\text{,}
\end{equation*}%
where this result is valid for any dimension $d\geq 1$.

The third spectral moment is proportional to the number of closed walks of
length three in the graph. We now derive an expression for the expected
number of triangular walks starting at a given node $v_{i}$ in a one-dimensional RGG. Since all nodes
are statistically equivalent, our result is valid for any other starting
node. For simplicity in our calculations, we consider that $v_{i}$ is
located at the origin. A triangular walk starting at node $v_{i}$ exists if and
only if there exist two nodes $v_{j}$ anv $v_{k}$ such that $\left\vert x_{j}\right\vert \leq r$, $\left\vert x_{k}\right\vert \leq r$, and $\left\vert x_{k}-x_{j}\right\vert \leq r$. Also, since the random distribution of
vertices on $\mathbb{T}^{1}$ is uniform (with density $n$), the probability
of nodes $v_{j}$ and $v_{k}$ being respectively located in the differential
lengths $\left[ x_{j}+dx_{j}\right) $ and $\left[ x_{k}+dx_{k}\right) $ is
equal to $n^{2}~dx_{j}dx_{k}$. Hence, one can compute the expected number of
triangular walks starting at node $v_{i}$ as%
\begin{equation*}
\mathbb{E}\left[ t_{i}\right] =\int \int_{\left( x_{j},x_{k}\right) \in
H_{2}\left( r\left( n\right) \right) }n^{2}~dx_{j}dx_{k},
\end{equation*}%
where%
\begin{align}
H_{2}\left( r\right) & =\left\{ \left( x_{j},x_{k}\right) \in \mathbb{T}^{2}%
\text{ s.t. }\left\vert x_{j}\right\vert \leq r,\right.   \label{Vol H2} \\
& \left. \text{ \ \ \ \ \ }\left\vert x_{k}-x_{j}\right\vert \leq
r,\left\vert x_{k}\right\vert \leq r\right\} .  \notag
\end{align}%
Thus, $\mathbb{E}\left[ t_{i}\right] $ can be computed as $n^{2}$Vol$\left[
H_{2}\left( r\left( n\right) \right) \right] $ (where Vol$\left( H\right) $
denotes the volume contained by the polyhedron $H$.) Notice that $H_{2}\left(
r\right) $ can be defined by a set of linear inequalies; hence, $H_{2}\left(
r\right) $ is a convex polyhedron that depends on $r$. Furthermore, the set
of linear inequalities in (\ref{Vol H2}) presents a homogeneous dependency
with respect to the parameter $r$. Therefore, we can write Vol$(H_{2}\left(
r\right) )$ as $r^{2}$Vol$(H_{2}\left( 1\right) )$. Finally, one can easily
compute the volume of $H_{2}\left( 1\right) $ to be equal to $3$. Thus, the
expected third spectral moment of a one-dimensional RGG is given by%
\begin{equation*}
\mathbb{E}\left[ m_{3}\right] =\frac{1}{n}\sum_{i=1}^{n}\mathbb{E}\left[
t_{i}\right] =3n^{2}r^{2}.
\end{equation*}

In the following, we extend the above technique to compute higher-order
expected spectral moments. Denote by $W_{i}^{\left( k\right) }$ the number
of closed walks of length $k$ starting at node $v_{i}$
in $G(\chi _{n};r\left( n\right) )$. Regarding $W_{i}^{\left( k\right) }$,
we derive the following result.

\begin{theorem}
\label{Spectral Moments 1D}The expected number of closed walks
of length $k$, $W_{i}^{\left( k\right) }$, in a random geometric graph, $%
G(\chi _{n};r)$, on $\mathbb{T}^{1}$ is given by%
\begin{equation*}
\mathbb{E}\left[ W_{i}^{\left( k\right) }\right] =\left( nr\right) ^{k-1}%
\frac{1}{2\left( k-1\right) !}\sum_{j=1}^{k-2}\binom{k-1}{j-1}~E_{k-1,j},
\end{equation*}%
where $E_{k-1,j}$ are the Eulerian numbers \footnote{%
The Eulerian number $E\left( n,k\right) $ gives the number of permutations
of $\{1,2,...,n\}$ having $k$ permutation ascents \cite{GKP94}.}.
\end{theorem}

\begin{proof}
Consider a particular closed walk, $\mathbf{w}_{k}=\left(
v_{1},v_{2},v_{3},...,v_{k},v_{1}\right) $, of length $k$ starting and
ending at node $v_{1}$ (which we locate at zero for computational
convenience). A walk $\mathbf{w}_{k}$ exists if and only if
there exists a set of $k-1$ nodes, $\left\{ v_{2},v_{3},...,v_{k}\right\} ,$
such that $\left\vert x_{1}\right\vert \leq r$, $\left\vert
x_{j+1}-x_{j}\right\vert \leq r$ for $j=2,...,k-1$, and $\left\vert
x_{k}\right\vert \leq r$. Since the distribution of vertices on $\mathbb{T}%
^{1}$ is uniform (with density $n$) one can compute the expectation of $%
W_{i}^{\left( k\right) }$ as 
\begin{equation*}
\mathbb{E}\left[ W_{i}^{\left( k\right) }\right] =\int_{(x_{2},...,x_{k})\in
H_{k-1}\left( r\left( n\right) \right) }n^{k-1}~dx_{2}...dx_{k},
\end{equation*}%
where%
\begin{align}
H_{k-1}\left( r\right) & =\left\{ (v_{2},v_{3},...,v_{k})\in \mathbb{T}^{k-1}%
\text{ s.t. }\left\vert v_{2}\right\vert \leq r,\right.   \label{Vol Hk-1} \\
& \left. \text{ \ \ }\left\vert x_{j+1}-x_{j}\right\vert \leq r\text{ for }%
j=2,...,k-1,\right.   \notag \\
& \left. \text{ \ \ \ }\left\vert x_{k}\right\vert \leq r\right\} .  \notag
\end{align}%
Thus, $\mathbb{E[}W_{i}^{\left( k\right) }]$ can be computed as $n^{k-1}$Vol$%
\left[ H_{k-1}\left( r\right) \right] $, where $H_{k-1}\left( r\right) $ is
a convex polyhedron defined by a set of linear inequalities. Finally, note
that the homogeneous structure of the system of linear inequalities defining 
$H_{k-1}\left( r\right) $ allows us to write Vol$(H_{k-1}\left( r\right)
)=r^{k-1}$Vol$(H_{k-1}\left( 1\right) )$. Therefore,%
\begin{equation}
\mathbb{E}\left[ W_{i}^{\left( k\right) }\right] =\left( nr\right) ^{k-1}%
\text{Vol}\left( H_{k-1}\left( 1\right) \right) .  \label{Walks as Volumes}
\end{equation}%
The volume of $H_{k-1}\left( 1\right) $ is a particular number, independent
of the RGG parameters, i.e., $n$ and $r$. Furthermore, we have found an
explicit analytical expression for the volume of $H_{k}\left( 1\right) $ for
any $k\geq 1$. Although we do not provide details of our derivation, due to
space limitations, an explicit expression for the volume of $H_{k}\left(
1\right) $ is given by \cite{PJ09}:%
\begin{equation}
\text{Vol}\left( H_{k}\left( 1\right) \right) =\frac{2}{k!}\sum_{j=1}^{k-1}%
\binom{k}{j-1}~E_{k,j},  \label{Volumes as Eulers}
\end{equation}%
where $E_{d,k}$ denotes the Eulerian numbers. Substituting (%
\ref{Volumes as Eulers}) in (\ref{Walks as Volumes}) we obtain the statement
of our lemma.
\end{proof}

In \cite{Las83}, Lasserre proposed an algorithm to compute the volume of a
polyhedron defined by a set of linear inequalities. We can use this algorithm to verify the validity of (\ref{Vol Hk-1}). Applying this algorithm to the set of inequalities in (\ref{Vol Hk-1}), we compute the following
volumes for $k=1,...,10$:%
\begin{align*}
H_{1}& =2,~H_{2}=3,~H_{3}=5.333...,~H_{4}=9.58333..., \\
H_{5}& =17.6000...,~H_{6}=32.70555...,~H_{7}=61.3587..., \\
H_{8}& =115.947...,~H_{9}=220.3238...,~H_{10}=420.825...
\end{align*}%
These numerical values match perfectly with our analytical expression in
Theorem \ref{Spectral Moments 1D}.

If $nr\left( n\right) =\Omega \left( \log n\right) $ (i.e., the average
degree grows as $\log n$, or faster), one can prove that $\mathbb{E}\left[ m_{k}\right] =\left(
1+O\left( \log ^{-1}n\right) \right) ~\mathbb{E[}W_{i}^{\left( k\right) }]$.
Hence, from (\ref{Walks as Volumes}) and (\ref{Volumes as Eulers}), we have
the following closed-form expression for the asymptotic expected spectral
moments:%
\begin{equation}
\mathbb{E}\left[ m_{k}\right] \asymp \left( nr\right) ^{k-1}\frac{1}{%
2\left( k-1\right) !}\sum_{j=1}^{k-2}\binom{k-1}{j-1}~E_{k-1,j}.
\label{Expected Spectral Moments 1D}
\end{equation}

In the following table
we compare the analytical result in (\ref{Expected Spectral Moments 1D})
with numerical realizations of the empirical spectral moments. In our
simulations, we distribute $n=1000$ nodes uniformly in $\mathbb{T}^{1}$ and
choose a connectivity radius $r=0.01$ (which results in an average degree $%
\mathbb{E}[d_{i}]=20$). The second, third, and forth column in the following
table represent the analytical expectations of the spectral moments, the
empirical average of the spectral moments from 10 random realizations of the
RGG, and the corresponding empirical typical deviation, respectively.%
\begin{equation*}
\begin{tabular}{|l|lll|}
\hline
$k$ & $\mathbb{E}\left[ m_{k}\right] $ & $\text{Empirical Average}$ & $\text{%
Typical Deviation}$ \\ \hline
1 & \multicolumn{1}{|c}{0} & \multicolumn{1}{c}{1.38e-16} & 
\multicolumn{1}{c|}{1.3e-15} \\ 
2 & \multicolumn{1}{|c}{20} & \multicolumn{1}{c}{19.9326} & 
\multicolumn{1}{c|}{0.0976} \\ 
3 & \multicolumn{1}{|c}{300} & \multicolumn{1}{c}{297.284} & 
\multicolumn{1}{c|}{4.3598} \\ 
4 & \multicolumn{1}{|c}{5,733} & \multicolumn{1}{c}{5,956.30} & 
\multicolumn{1}{c|}{196.94} \\ \hline
\end{tabular}%
\end{equation*}%
Our numerical results present an excellent match with our analytical
predictions.

\subsection{Spectral Moments of Two-Dimensional RGG's}

In this subsection, we derive expressions for the first three
expected spectral moments of $G(\chi _{n};r\left( n\right) )$ when the nodes
are uniformly distributed in $\mathbb{T}^{2}$. The expressions for
the first and second expected spectral moments are $m_{1}=0$ and $\mathbb{E}\left[
m_{2}\right] =\pi nr^{2}$. The third spectral moment is
proportional to the number of closed walks of length three in the graph. In
the two-dimensional case, we count the number of triangular walks using a
technique that we illustrate in Fig. $\emph{1}$. In this figure, we plot two
nodes $v_{i}$ and $v_{j}$. The
parameters $\rho $ and $\phi $ in Fig. $\emph{1}$ denote the distance and
angle between these two nodes, i.e., $\rho \triangleq \left\Vert \mathbf{x}%
_{j}-\mathbf{x}_{i}\right\Vert $ and $\phi =\measuredangle \left( \mathbf{x}%
_{j}-\mathbf{x}_{i}\right) $. An edge between $v_{i}$ and $v_{j}$ exists if an
only if $v_{j}$ is located inside the circle $S_{i}\left( r\right) $. In this setting, the probability
of existence of a triangle touching both $v_{i}$ and $v_{j}$ is equal to the probability of a third node $v_{k}$ being in the
shaded area $A_{l}$ (see Fig. $\emph{1}$). This area is the result of
intersecting the circles $S_{i}\left( r\right) $ and $S_{j}\left( r\right) $%
, and the resulting probability is equal to $n~A_{l}$. The
intersecting region $A_{l}$ is a symmetric lens which area can be computed
as a function of $\rho $ and $r$ as follows:%
\begin{equation}
A_{l}\left( \rho ;r\right) =\left\{ 
\begin{array}{cc}
2r^{2}\cos ^{-1}\left( \frac{\rho }{2r}\right) -\frac{\rho }{2}\sqrt{%
4r^{2}-\rho ^{2}}, & \text{for }\rho \leq r, \\ 
0, & \text{for }\rho >r.%
\end{array}%
\right.   \label{Lens Area}
\end{equation}%
Therefore, we can compute the expected number of
triangles by integrating over the set of all possible positions of $v_{j}$,
i.e., $\eta \in \left[ 0,r\right] $ and $\phi \in \lbrack 0,2\pi )$, as
follows%
\begin{equation}
\mathbb{E}\left[ t_{i}\right] =\int_{\rho =0}^{r}\int_{\phi =0}^{2\pi
}n^{2}A_{l}\left( \rho ;r\right) ~\rho ~d\rho ~d\phi .
\label{Triangle Integral}
\end{equation}%
After substituting (\ref{Lens Area}) in (\ref{Triangle Integral}), we can
explicitly solve the resulting integral to be%
\begin{equation}
\mathbb{E}\left[ t_{i}\right] =\left( \pi -\frac{3\sqrt{3}}{4}\right) \pi
\left( nr^{2}\right) ^{2}\approx 5.78\left( nr^{2}\right) ^{2}.
\label{Triangles 2D}
\end{equation}%
Consequently, we have the following expression for the third
expected spectral moment $\mathbb{E}\left[ m_{3}\right] =\frac{1}{n}%
\sum_{i=1}^{n}\mathbb{E}\left[ t_{i}\right] =\mathbb{E}\left[ t_{i}\right] $.

\begin{figure}
 \centering
 \includegraphics[width=0.85\linewidth]{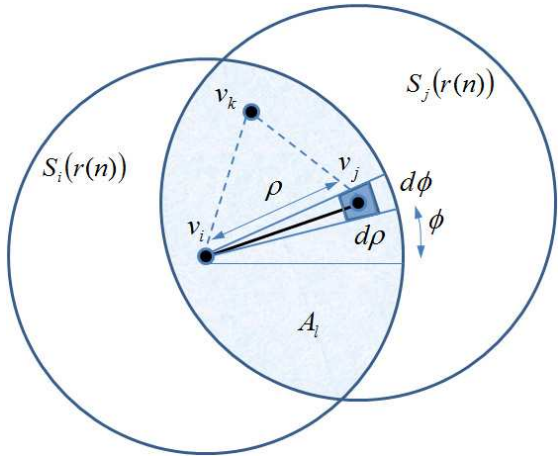}
 \caption{This figure illustrates the technique proposed in Section III.B to count the number of triangular walks in a two-dimensional RGG.}
% \label{fig:XXXXX}
\end{figure}

In the following, we extend the technique introduced above to compute closed walks of arbitrary length. Denote by $W_{i}^{\left(
k\right) }$ the number of closed walks of length $k$ starting at node $v_{1}$
in $G(\chi _{n};r\left( n\right) )$. The idea behind our technique is
illustrated in Fig. $\emph{2}$, where we represent a particular closed
walk of length $6$. We denote this walk by\ $\mathbf{w}%
_{k}=\left( v_{1},v_{2},...,v_{k-1},v_{k},v_{1}\right) $. We define the
following set of relative distances and angles between every pair of
connected vertices: $\rho_{i}\triangleq \left\Vert \mathbf{x}_{i+1}-\mathbf{x}%
_{i}\right\Vert $ and $\phi _{i}=\measuredangle \left( \mathbf{x}_{i+1}-%
\mathbf{x}_{i}\right) $ for $i=1,...,k-2$. We also define the following parameter%
\begin{equation}
\rho =\left\vert \sum_{j=1}^{k-2}\rho_{j}e^{\mathbf{i}\alpha _{j}}\right\vert ,
\label{Rho parameter}
\end{equation}%
($\mathbf{i}=\sqrt{-1}$) which is the resulting distance between nodes $v_{k-1}
$ and $v_{1}$ given a particular set of distances and angles $\left\{ \left(
r_{i},\phi _{i}\right) \right\} _{i=1,...,k-2}$ (see Fig. $\emph{2}$). In
this setting, the conditional probability of existence of a walk $\mathbf{w}%
_{k}=\left( v_{1},v_{2},...,v_{k-1},v_{k},v_{1}\right) $ given the set of
relative positions, $\left\{ \left( r_{i},\phi _{i}\right) \right\}
_{i=1,...,k-2}$, is equal to the probability of $v_{k}$ being in the shaded
area $A_{l}$ in Fig. $\emph{2}$. We have an expression for this area in (\ref%
{Lens Area}), where $\rho $ is defined in (\ref{Rho parameter}).
Finally, we can compute the expectation of $W_{i}^{\left( k\right) }$ by
performing an integration over the set of all possible positions (i.e., $%
\rho _{j}\in \left[ 0,r\right] $ and $\phi _{j}\in \lbrack 0,2\pi )$ for $%
j=2,...,k-1$), as follows%
\begin{equation*}
\mathbb{E}\left[ W_{i}^{\left( k\right) }\right] =n^{k-1}\int_{\left( 
\mathbf{\eta ,\varphi }\right) \in C_{k-2}}A_{l}\left( \rho ;r\right)
~\prod_{j=2}^{k-1}\eta _{j}~d\mathbf{\eta ~}d\mathbf{\varphi ,}
\end{equation*}%
where $\mathbf{\eta =}\left( \rho _{2},...,\rho _{k-1}\right) $, $\mathbf{%
\varphi }=\left( \phi _{2},...,\phi _{k-1}\right) $, and $%
C_{k-2}=\{\left( \mathbf{\eta ,\varphi }\right) :\mathbf{\eta \in }\left[ 0,r%
\right] ^{k-2}$ and $\mathbf{\varphi }\in \lbrack 0,2\pi )^{k-2}\}$.
Although a closed-form for the above expression can only be computed for $%
k \leq 3$, we can always find a good approximation via numerical integration.
For example, the integration for $k=4$ gives us $\mathbb{E[}W_{i}^{\left(
4\right) }]\approx 14.2511\left( nr^{2}\right) ^{3}$.

\begin{figure}
 \centering
 \includegraphics[width=0.95\linewidth]{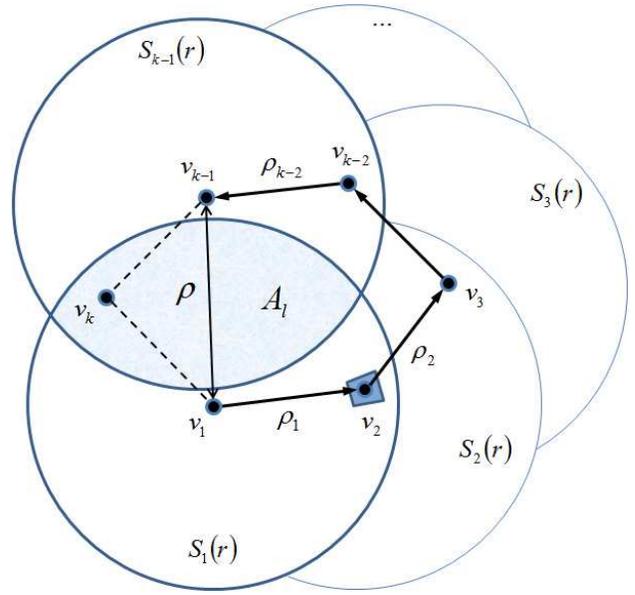}
 \caption{This figure illustrates the technique proposed in Section III.B to count the number of closed walks of length $k$ in a two-dimensional RGG.}
% \label{fig:XXXXX}
\end{figure}

In the following table, we compare our analytical results with numerical
realizations of the empirical spectral moments of a two-dimensional RGG. In
our simulations, we distribute $n=1000$ nodes uniformly on $\mathbb{T}^{2}$
and choose a connectivity radius $r=\sqrt{50/\pi n}\approx 0.1784$ (which
results in an average degree $\mathbb{E}[d_{i}]=50$). The second, third, and
forth columns in the following table represent the analytical expectation of
the spectral moments, the empirical average from 10 random realizations, and
the corresponding empirical typical deviation, respectively.%
\begin{equation*}
\begin{tabular}{|l|lll|}
\hline
$k$ & $\mathbb{E}\left[ m_{k}\right] $ & $\text{Empirical Average}$ & $\text{%
Typical Deviation}$ \\ \hline
1 & \multicolumn{1}{|c}{0} & \multicolumn{1}{c}{-9.2e-16} & 
\multicolumn{1}{c|}{1.1e-15} \\ 
2 & \multicolumn{1}{|c}{50} & \multicolumn{1}{c}{50.0820} & 
\multicolumn{1}{c|}{0.3908} \\ 
3 & \multicolumn{1}{|c}{1,464.1} & \multicolumn{1}{c}{1,475.8} & 
\multicolumn{1}{c|}{37.3777} \\ 
4 & \multicolumn{1}{|c}{59,452} & \multicolumn{1}{c}{60,127} & 
\multicolumn{1}{c|}{2,955.3} \\ \hline
\end{tabular}%
\end{equation*}%
Our numerical results present an excellent match with our analytical
predictions.

In the following section, we use the results introduced in this section to
study the spreading of an infection in a random geometric network.

\section{Spectral Analysis of Virus Spreading}

In this section, we use the expressions for the expected spectral moments to
design random geometric networks to tame an initial viral infection in the
network. In our design problem, we consider that the size of the network $n$
and the parameters in (\ref{Epidemic Model}), i.e., $\beta $ and $\delta $,
are given. Hence, our design problem is reduced to studying the range of
values of $r$ for which the RGG is well-suited to tame an initial viral
infection.

A sufficient condition for local stability around the disease-free state was
given in (\ref{Epidemic Conditions}). Thus, we have to find the range of
values of $r$ for which the associated spectral radius $\lambda _{\max }$ is
smaller than the ratio $\delta /\beta $. In the following subsection, we
show how to derive an analytical upper bound for the spectral radius based
on the expected spectral moments.

\subsection{Analytical Upper Bound for the Spectral Radius}

In order to upper-bound the spectral radius, we use Wigner's high-order
moment method \cite{Wig58}. This method provides a probabilistic upper bound
based on the asymptotic behavior of the $k$-th expected spectral moments for
large $k$. We present the details for a one-dimensional RGG, although the
same technique can be applied to RGG's in higher dimensions. For a
one-dimensional RGG in the connectivity regime, we derived an explicit
expression for the expected spectral moments in (\ref{Expected Spectral
Moments 1D}). A logarithmic plot of Vol$\left( H_{k}\left( 1\right) \right) $ for $k=1,2,...,9$ unveils that Vol$\left( H_{k}\left( 1\right) \right) \rightarrow \beta_{1}c_{1}^{k}$ for large-order moments (a line in logarithmic
scale), where, from a numerical fitting, we find that $\beta _{1}=0.35$ and $%
c_{1}=1.9192$. Therefore, from (\ref{Expected Spectral Moments 1D}) we have%
\begin{equation*}
\mathbb{E}\left[ m_{k}\right] \asymp \beta _{1}\left( c_{1}nr\right)
^{k},
\end{equation*}%
for large $k$.

%\begin{figure}
% \centering
% \includegraphics[width=0.9\linewidth]{Fig3.eps}
% \caption{A logarithmic plot of Vol$%
%\left( H_{k}\left( 1\right) \right) $ for $k=1,2,...,9$ unveils that Vol$\left( H_{k}\left( 1\right) \right) \rightarrow \beta
%_{1}c_{1}^{k}$ for large-order moments (a line in logarithmic
%scale).}
%% \label{fig:XXXXX}
%\end{figure}

For even-order expected spectral moments (i.e., $k=2s$ for $s\in \mathbb{N}$%
), the following holds%
\begin{equation*}
\mathbb{E}\left[ m_{2s}\right] =\frac{1}{n}\sum_{i=1}^{n}\mathbb{E}[\lambda
_{i}^{2s}]\geq \frac{1}{n}\mathbb{E}[\lambda _{\max }^{2s}].
\end{equation*}%
Define $f\left( n\right) =n^{1-\delta }\log n$; thus, for any $\varepsilon
,\delta >0$ (and $c_{1}=1.9192$), we can apply Markov%
%TCIMACRO{\U{b4}}%
%BeginExpansion
\'{}%
%EndExpansion
s inequality as follows%
\begin{eqnarray*}
\mathbb{P}\left( \lambda _{\max }^{2s}\geq (c_{1}nr+\varepsilon rf\left(
n\right) )^{2s}\right)  &\leq &\frac{\mathbb{E}[\lambda _{\max }^{2s}]}{%
(c_{1}nr+\varepsilon rf\left( n\right) )^{2s}} \\
&\leq &\frac{n~\mathbb{E}\left[ m_{2s}\right] }{(c_{1}nr+\varepsilon
rf\left( n\right) )^{2s}},
\end{eqnarray*}%
For large $s$, one can prove that \cite{Pre07}%
\begin{equation*}
\mathbb{P}\left( \lambda _{\max }\geq c_{1}nr+\varepsilon rf\left( n\right)
\right)  \leq n\beta _{1}\exp \left( -\frac{\varepsilon }{c_{1}}sn^{-\delta }\log
n\right) .
\end{equation*}%
Assuming that $s$ grows as $\beta _{2}n^{\delta }$, for $\beta _{2},\delta >0
$, we have%
\begin{equation*}
\mathbb{P}\left( \lambda _{\max }\geq c_{1}nr+\varepsilon rf\left( n\right)
\right) \leq n\beta _{1}\exp \left( -\frac{\beta _{2}\varepsilon }{c_{1}}%
\log n\right) =o\left( 1\right) ,
\end{equation*}%
for all sufficiently large $\varepsilon $. Thus,%
\begin{equation}
\lim_{n\rightarrow \infty }\mathbb{P}\left( \lambda _{\max
}<c_{1}nr+\varepsilon rn^{1-\delta }\log n\right) =1.  \label{Upper Bound}
\end{equation}%
In other words, $\lambda _{\max }$ is upper-bounded by $cnr+\varepsilon
rn^{1-\delta }\log n$ with probability $1$ for $n\rightarrow \infty $. In
practice, for a large (but finite) $n$, we can use $1.9192~nr$ as an upper
bound of $\lambda _{\max }$. In Fig. 4, we plot the empirical
spectral radius of an RGG with $n=1000$ and $r\left( n\right) =\bar{d}/2n$,
with expected degrees $\bar{d}=[$10:1:100$]$ (circles in the figure). We also plot the values of our
analytical upper bound, $1.9192~nr$, in solid line.

\begin{figure}
 \centering
 \includegraphics[width=0.9\linewidth]{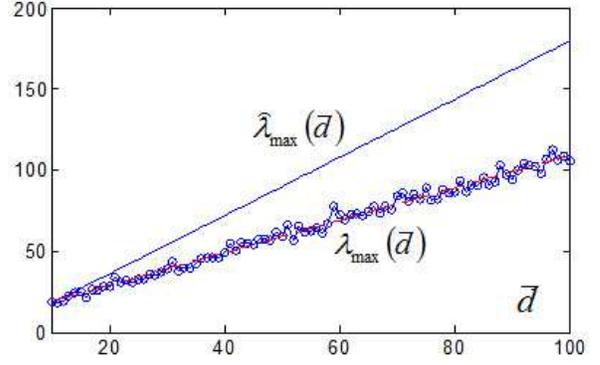}
 \caption{Comparison between the empirical spectral radius of an RGG (circles in the plot) and the values of our analytical upper bound (solid line) for $n=1000$ and $r\left( n\right) =\bar{d}/2n$, with expected degrees $\bar{d}=[$10:1:100$]$.}
% \label{fig:XXXXX}
\end{figure}

The technique introduced in this subsection is also valid for RGG's in
higher-dimensions. In general, one can prove that for a $d$-dimensional RGG that the expected spectral moment grows as $\mathbb{E}\left[ m_{k}\right]
\rightarrow \beta _{d}\left( c_{d}nr^{d}\right) ^{k}$. Applying Wigner's
high-order moment method to this sequence, one can derive a probabilistic
upper bound similar to (\ref{Upper Bound}). In particular, we have that $\lambda _{\max }<c_{d}nr^{d}$ for large $n
$ with high probability. In the following subsection, we use our results to
design the connectivity radius of an RGG in order to tame an initial viral
infection.

\subsection{Spectral Radius Design}

Once the spectral radius is upper-bounded, our design problem becomes
trivial. Since (\ref{Epidemic Conditions}) represents a sufficient condition
for local stability around the disease-free state, we have the following
condition to tame an initial viral infection for a $d$-dimensional RGG:%
\begin{equation*}
\lambda _{\max }\left( G(\chi _{n};r\right) )<c_{d}nr^{d}<\frac{\delta }{%
\beta },
\end{equation*}%
which implies the following design condition for the connectivity radius:%
\begin{equation}
r<\left( \frac{\delta }{\beta c_{d}n}\right) ^{1/d},  \label{Radius design}
\end{equation}%
where $c_{d}$ is a positive constant that depends on the dimension of $%
\mathbb{T}^{d}$. For example, in the one-dimensional case, we have $%
c_{1}=1.9192$; hence, (\ref{Radius design}) becomes $r<\delta /\left(
1.9192~\beta n\right) $. We now validate this result with several numerical
simulations of a viral infection in a one-dimensional RGG.

Consider an RGG with $n=1000$ nodes and a connectivity radius of $r=0.005$
(which implies an average degree of $10$). The resulting spectral radius in
this RGG is $\lambda _{\max }=17.2629$. In our numerical simulations, we
choose the initial probability of infection to be $p_{i}\left[ 0\right] \sim
0.01$\textsf{Unif}$[0,1]$; hence, approximately $1\%$ of the nodes in the
network are initially infected. In our first experiment, we choose a rate of
infection $\beta =0.020$, and a recovery rate $\delta =0.018$. Since the
sufficient condition for viral control in (\ref{Radius design}) is not
satisfied, we cannot guarantee an initial infection to be tamed. In
Fig. $\emph{5}$ we show an image of the evolution of the probability of
infection for this case. This figure is a color map for the simultaneos
evolution of $p_{i}\left[ n\right] $ for $i=1,...,1000$. Each horizontal
line represents the value of $p_{i}\left[ n\right] $ for a particular $i$. In
this color map, blue represents a zero value, green and yellow tones
represent intermediate values, and red represents values close to one. On the other hand,
if we increase the recovery rate to $\delta =0.35$ keeping the rest of
parameters fixed, we have that $\delta /\beta =17.50>\lambda _{\max }$ and
we satisfy condition (\ref{Epidemic Conditions}). Hence, the probability of
infection of every node is guaranteed to converge towards zero. In Fig. $%
\emph{6}$, we observe the color map for the evolution of the probability of
infection in this case, where we clearly observe how $p_{i}\left[ n\right]
\rightarrow 0$ for all $i$. Hence, this latter RGG is well-suited to tame
initial viral infections.

\begin{figure}
 \centering
 \includegraphics[width=0.9\linewidth]{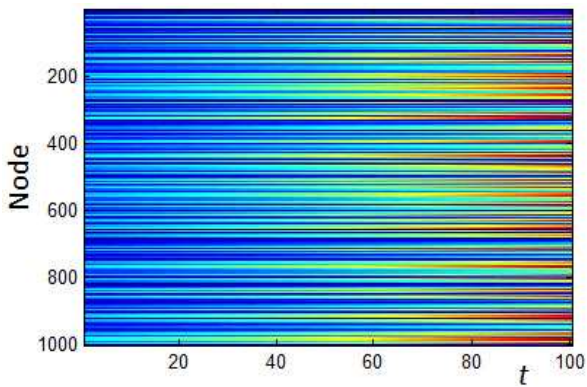}
 \caption{Color map representing the evolution of the probabilities of infection $p_{i}\left[ n\right] $ for $i=1,...,1000$ in an RGG with $n=1000$ nodes, connectivity radius $r=0.005$, rate of infection $\beta =0.020$, and recovery rate $\delta =0.018$. Each horizontal line represents the value of $p_{i}\left[ n\right] $ for a particular $i$. In this color map, blue represents a zero value, green and yellow tones represent intermediate values, and red represents values close to one. In this case, we observe an epidemic outbreak.}
% \label{fig:XXXXX}
\end{figure}

\begin{figure}
 \centering
 \includegraphics[width=0.9\linewidth]{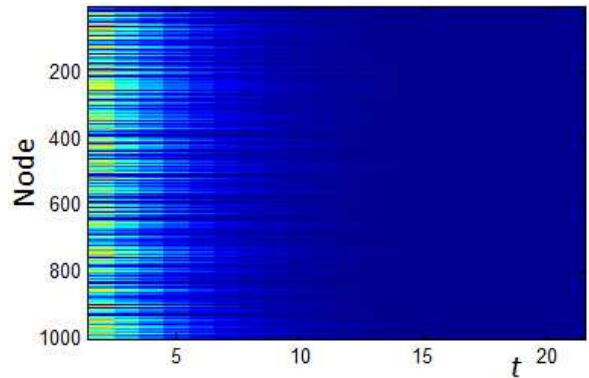}
 \caption{Color map representing the evolution of the probabilities of infection $p_{i}\left[ n\right]$ when we increase the recovery rate to $\delta =0.35$ (the rest of parameters are the same as we used for Fig. 5). We observe how the probability of infection of every node converges towards zero in this case. }
% \label{fig:XXXXX}
\end{figure}

\section{Conclusions}
In this paper, we have studied the spreading of a viral infection in a random geometric graph from a spectral point of view. We have focused our attention on studying the eigenvalue distribution of the adjacency matrix. We have derived, for the first time, explicit expressions for the spectral moments of the adjacency matrix as a function of the density of nodes and the connectivity radius. We have then applied our results to the problem of viral spreading in a network with a low-density infection. Using our expressions, we have derived upper bounds for the spectral radius of the adjacency matrix. Finally, we have applied this upper bound to design random geometric graphs that are well-suited to tame an initial low-density infection. Our numerical results match our predictions with high accuracy.

%%%%%%%%%%%%%%%%%%%%%%%%%%%%%%%%%%%%%%%%%%%%%%%%%%%%%%%%%%%%%%%%%%%%%%%%%%%%%%%%

\end{document}